\documentclass{article}

\usepackage{\jobname}
\usepackage[utf8]{inputenc}

\title{%
Complexity Classifications for logic-based Argumentation%
\footnote{A preliminary version of this work appeared in the Proceedings of
COMMA  2012,
 \emph{Frontiers in Artificial Intelligence and Applications}, Volume 245,
                     pages 237--248,  2012.}}
\author{%
NADIA CREIGNOU\\
Aix-Marseille Universit\'e \and
UWE EGLY\footnote{This work was partially supported by the Austrian Science Foundation (FWF) under grant S11409-N23}\\
Technische Universit{\"a}t Wien \and
JOHANNES SCHMIDT\footnote{Supported by the National Graduate School in Computer Science (CUGS), Sweden.}\\
Link{\"o}ping University
}

\begin{document}

\maketitle

\begin{abstract}
We consider logic-based argumentation in which an argument is a pair
$(\Phi,\alpha)$, where the support $\Phi$ is a minimal consistent set of
formul{\ae} taken from a given knowledge base (usually denoted by $\Delta$) that entails the claim
$\alpha$ (a formula).
We study the complexity of three central problems in argumentation:
the existence of a support $\Phi\subseteq\Delta$,
the verification of a support and 
the relevance problem (given $\psi$ is there a support $\Phi$ such that $\psi
\in \Phi$?).
When arguments are given in the full language of propositional logic these problems
are computationally costly tasks: the verification problem is $\DP$-complete, the others are $\SigPtwo$-complete. We study these problems in Schaefer's famous framework where the considered propositional formul\ae\
are in generalized conjunctive normal form. This means that formul{\ae} are conjunctions of constraints build upon a fixed
finite set of Boolean relations $\Gamma$ (the constraint language).
We show that according to the properties of this language $\Gamma$, deciding
whether there exists a support for a claim in a given knowledge base is either
polynomial, $\NP$-complete, $\coNP$-complete or $\SigPtwo$-complete. 
We present  a dichotomous classification, $\P$ or $\DP$-complete, for the
verification problem and a trichotomous classification for the relevance problem into
either polynomial, $\NP$-complete, or $\SigPtwo$-complete. These last two
classifications are obtained by means of algebraic tools.\end{abstract}

\section{Introduction}\label{sec:introduction}
Argumentation can be seen as a generalization of many forms of nonmonotonic reasoning previously developed \cite{Dung95}.
It is nowadays a very active research area in artificial intelligence. One can identify, among others, two important lines of research:
\emph{abstract} argumentation \cite{Dung95} and \emph{logic-based} (or deductive) argumentation \cite{behu01,CML00,PV02,BH08}.
The former focuses on the relations between arguments based on the property of arguments to attack others, thereby ignoring
the internal structure of an argument and the nature of the attack relation.
In this work we explore logic-based argumentation  in which an
argument is a pair $(\Phi,\alpha)$, where the support $\Phi$ is a minimal consistent set of formul{\ae} that entails the claim $\alpha$ (a formula). 

From a complexity theoretic viewpoint, computing the support of an argument
is a very hard problem. Indeed, in the full language of
propositional logic, given a knowledge base $\Delta$, the problem of deciding
whether there exists a support $\Phi\subseteq \Delta$ for a given
claim $\alpha$ has 
been shown to be \SigPtwo-complete \cite{pawoam03}.
Since this problem underlies many reasoning 
problems in logic-based argumentation, like for instance the computation 
of argument trees as proposed by Besnard and Hunter \cite{behu01}, it is
natural to try to identify fragments of propositional logic for which 
 the deduction problem is easier.

A first step towards an extensive study of the complexity of argumentation in
fragments of propositional logic  was taken 
in \cite{CrScThWo11} in Post's framework, where the authors considered formul\ae\ built upon a
restricted set of connectives. They obtained a full classification of  
various argumentation problems depending on the set of allowed connectives. A
similar yet different approach is not to restrict the connectives but to
restrict the syntactic shape of the formul\ae. This refers to the 
well-known Schaefer's framework in which formul\ae\ are considered in
generalized conjunctive normal form with clauses formed upon a fixed set of
relations $\Gamma$ (the  constraint language). Such formul\ae\  are
called \emph{$\Gamma$-formul\ae}. This  framework captures well-known
classes of formul\ae\  in conjunctive normal form, e.g., Horn, definite Horn
or 2-CNF.  A wide range of algorithmic problems have been studied in this
context (for a survey see \cite{crvo08}),
and in particular the abduction problem \cite{CrZa06,NoZa08}. Preliminary results concerning
argumentation have been obtained in \cite{CreignouES12}.

Our main contribution is a systematic complexity classification
for the problems of existence $\ARG$, verification $\ARGCHECK$ and  relevance $\ARGREL$ in
terms of all
possible sets of relations $\Gamma$.
These problems are formally defined in Section \ref{sec:argumentation_def}. They
can be described as follows:

\begin{itemize}\setlength{\itemsep}{\itemseplength}

\item $\ARG$: given $(\Delta$, $\alpha)$, does there exist a support $\Phi\subseteq \Delta$ for $\alpha$ ?

\item $\ARGCHECK$: given $(\Phi,\alpha)$, is it an argument?

\item $\ARGREL$: given $(\Delta,\alpha, \psi)$, is there a support $\Phi \subseteq \Delta$  for $\alpha$ such that $\psi \in \Phi$?


\end{itemize}
We prove that depending on the set of allowed
relations $\Gamma$ in our formul\ae\ in generalized conjunctive normal form,
deciding the existence of a support is either in $\P$, or $\NP$-complete,
or $\coNP$-complete or $\SigPtwo$-complete. The verification problem $\ARGCHECK$ is
either in $\P$ or $\DP$-complete, whereas the relevance problem $\ARGREL$ obtains a trichotomous classification into membership in $\P$, or $\NP$-complete, or $\SigPtwo$-complete.

For many classifications obtained in Schaefer's framework the so-called \emph{algebraic approach} turned out to be applicable.
Roughly speaking this means that the complexity of a problem parametrized by a constraint
language $\Gamma$ is fully determined by its ``expressive power'', defined by algebraic closure properties (this will be made precise in the following).
In the case of the argumentation problems we consider, it is however not clear
how to prove such a statement on the complexity. We therefore develop some new techniques that
still allow us to use parts of these elegant algebraic tools. While in the case
of $\ARG$ and $\ARGCHECK$ we finally obtain that their complexity is indeed characterized by the expressive power of the constraints, we show that in the case of $\ARGREL$ the usual
algebraic approach is definitely not applicable (unless $\P = \NP$): we identify
constraint languages $\Gamma_1, \Gamma_2$  having the same expressive power such that $\ARGREL(\Gamma_1)$
is in $\P$ and $\ARGREL(\Gamma_2)$ is $\NP$-complete.


The paper is organized as follows. In Section \ref{sec:preliminaries} we give
some basics  on complexity theory, we present  Schaefer's framework and we 
remind some complexity classifications that will be of use in our proofs
(in particular we explain how our work relates to the complexity
classifications obtained for abduction). In Section \ref{sec:argumentation_def}
we define formally the problems we are interested in. In Section
\ref{sec:tools} we present the algebraic tools we will use and give a series of
technical lemmas. In the following sections we establish complexity
classifications for the existence (Section \ref{sec:complexity_arg}), verification
(Section \ref{sec:complexity_argcheck}) and relevance  (Section
\ref{sec:complexity_argrel}) problems. We conclude in Section
\ref{sec:conclusion}.

\section{Preliminaries}\label{sec:preliminaries}

We assume familiarity with the syntax and semantics of propositional and first order logic. 
A \textit{literal} is a variable (\textit{positive literal})  or its negation (\textit{negative literal}), a ($k$-) \textit{clause} is a disjunction of ($k$) literals and a formula in ($k$-) \textit{CNF} is a conjunction of ($k$-) clauses. 
A formula in CNF is \textit{Horn} (resp., \textit{dual Horn}) if every clause contains at most one positive  (resp., negative) literal. 
A formula in CNF is\textit{ positive}  (resp., \textit{negative}) if every clause contains   positive  (resp., negative) literals only. 
\subsection{Complexity theory}

We require standard notions of complexity theory.
For the   problems studied in the paper the arising complexity degrees
encompass the classes
$\P$, $\NP$, $\coNP$, $\DP$ and $\SigPtwo$,
where $\DP$ is defined as the set of languages 
recognizable by the difference of two languages in $\NP$, 
i.e., $\DP := \{ L_1\setminus L_2 \mid L_1,L_2 \in \NP\} =  \{ L_1 \cap
L_2 \mid L_1 \in \NP, L_2 \in \co\NP\}$, and $\SigPtwo$ is the set of languages
recognizable by
nondeterministic polynomial-time Turing machines with an $\NP$ oracle.
For our hardness results we employ \emph{logspace many-one reductions},
defined as follows:
a language $A$ is logspace many-one reducible to some language $B$ (written $A
\leqlogm B$) if there exists a logspace-computable function $f$ such that $x \in A$ if and
only if  $f(x) \in B$.
For more background information on complexity theory, the reader is referred
to \cite{pap94}. 
We will use, among others, the following standard problems to prove hardness results.


  \problemdef{$\ThreeSAT$ ($\NP$-complete according to \cite{Cook71})}
      {A propositional formula $\varphi$ in 3-CNF.}
      {Is $\varphi$ satisfiable?}

  \problemdef{$\PosOneThreeSAT$ ($\NP$-complete according to \cite{Schaefer78})}
      {A propositional formula $\varphi$ in 3-CNF with only positive literals.}
      {Is there an assignment to the variables of $\varphi$ that sets in each clause exactly one variable to true?}

  \problemdef{$\CRITICALSAT$ ($\DP$-complete according to \cite{pawo88})}
      {A propositional formula $\varphi$ in 3-CNF.}
      {Is $\varphi$ unsatisfiable but removing any of its clauses makes it satisfiable?}

\subsection{Constraint languages and $\Gamma$-formul\ae}

A \emph{logical relation} of arity $k$ is a relation $R\subseteq\{0,1\}^k$. In
this paper we will only consider nontrivial  relations, i.e., $R\ne\emptyset$
and $R\ne\{0,1\}^k$. By abuse of notation we do
not make a difference between a relation and its predicate symbol.
We will use $\T$ and $\F$ as the two unary constant relations $\T = \{1\}$ and $\F =
\{0\}$.
A \emph{constraint}, $C$, is a formula $C=R(x_1,\dots,x_k)$,
where $R$ is a logical relation of arity $k$ and the $x_i$'s are (not
necessarily distinct) variables.
For instance the two constraints $\T(x)$ and $\F(x)$ stand for the two unit 
clauses $(x)$ and $(\neg x)$, respectively. An \emph{$R$-constraint} is a constraint built in using the relation $R$.
If $u$ and $v$ are two variables, then $C[v/u]$
denotes the constraint obtained from $C$ by replacing 
each occurrence of $v$ by $u$. If $V$ is a set of variables, then $C[V/u]$
denotes the result of substituting $u$ to every occurrence of every variable of
$V$ in $C$. An assignment $m$ of truth values to the variables \emph{satisfies}
the constraint $C$ if $\bigl(m(x_1),\dots,m(x_k)\bigr)\in R$. A \emph{constraint
language} $\Gamma$ is a finite set of  nontrivial logical relations.
A \emph{$\Gamma$-formula} $\phi$ is a conjunction of constraints using  only
logical relations from $\Gamma$ and is hence a quantifier-free first-order
formula. With $\var(\phi)$ we denote the set of (free) variables appearing in $\phi$. A
$\Gamma$-formula $\phi$ is satisfied by an assignment $m:\var(\phi)\to\{0,1\}$
if $m$ satisfies all constraints in $\phi$ simultaneously (such a satisfying
assignment is also called a \emph{model} of $\phi$). Assuming a canonical order
on the variables we can regard models as tuples in the obvious way and we do not
distinguish between a formula $\phi$ and the logical relation $R_\phi$ it
defines, i.e., the relation consisting of all models of $\phi$.
We say that two first-order formul\ae\  $\varphi$ and $\psi$ are equivalent,
$\varphi\equiv\psi$, if every assignment $m: \Vars{\varphi} \cup \Vars{\psi}
\rightarrow \{0,1\}$
on the combined variable sets satisfies $\varphi$ if and only if it satisfies $\psi$.
We write $\varphi \models \psi$ if $\varphi$ entails  $\psi$, i.e., if $\psi$ is satisfied by any assignment
$m: \Vars{\varphi} \cup \Vars{\psi} \rightarrow \{0,1\}$ that satisfies
$\varphi$.

\medskip 

Throughout the text we refer to different types of Boolean relations following
Schaefer's terminology \cite{Schaefer78}. We say that a Boolean relation~$R$ is
\begin{itemize}\setlength{\itemsep}{\itemseplength}
\item \emph{Horn} (resp. \emph{dualHorn}) if $R$ can be defined by a CNF formula which is $\horn$ (resp. $\dualhorn$);
\item \emph{bijunctive} if it can be defined by a 2-CNF formula;
\item \emph{affine} if it can be defined by an $\affine$ formula, i.e., a conjunction of XOR-clauses
			(consisting of an XOR of some variables plus maybe the constant 1) --- such a formula may also be seen as a system of linear equations over GF$[2]$;
\item \emph{positive} (resp. \emph{negative}) if $R$ can be defined by a
  positive (resp. negative) CNF formula;
\item \emph{0-valid} (resp., \emph{1-valid}) if $R(0,\ldots, 0)=1$ (resp., $R(1,\ldots, 1)=1$);
\item \emph{$\varepsilon$-valid} if $R$ is either 0-valid, or 1-valid or both;
\item \emph{complementive} if, for all $m\in R$, we have also $\overline{m}\in R$, where $\overline{m}$ denotes the dual assignment of $m$ defined by $\overline{m}(x)=1-m(x)$. 
\end{itemize}

Finally a constraint language
$\Gamma$ is $\horn$ (resp. $\dualhorn$, $\bijunctive$, $\affine$, $\positive$, $\negative$, $\zerovalid$, $\onevalid$, $\epsilonvalid$, $\complementive$) if every relation in $\Gamma$ is $\horn$ (resp. $\dualhorn$, $\bijunctive$, $\affine$, $\positive$, $\negative$, $\zerovalid$, $\onevalid$, $\epsilonvalid$, $\complementive$). We say that a constraint language is \emph{Schaefer} if $\Gamma$ is either $\horn$, $\dualhorn$, $\bijunctive$, or $\affine$. 

There exist easy criteria to determine if a given relation is $\horn$, $\dualhorn$, $\bijunctive$,
or $\affine$. Indeed all these classes can be characterized by their polymorphisms
(see e.g., \cite{crvo08} for a detailed description). We recall here the characterizations for $\horn$ and $\dualhorn$.
The binary operations of conjunction and disjunction applied on $k$-ary Boolean vectors are applied coordinate-wise.
\begin{itemize}\setlength{\itemsep}{\itemseplength}
	\item $R$ is $\horn$ if and only if $m, m'\in R$ implies $m \land m' \in R$.
	\item $R$ is $\dualhorn$ if and only if $m, m'\in R$ implies $m \lor m' \in R$.
\end{itemize}


\subsection{Related complexity classifications}

The formul\ae\ in  generalized conjunctive
normal form, $\Gamma$-formul\ae, defined as in the section above,  have provided
a rich framework to obtain
complexity classifications for computational problems involving Boolean
formul\ae\ (see e.g., \cite{crvo08}). We recall here some of them that will be of
use in the following. Moreover we make clear the relationship between the
complexity of argumentation and the  complexity of abduction.  
\smallskip

The satisfiability problem for $\Gamma$-formul{\ae}, denoted by $\SAT(\Gamma)$,
was first studied by Schaefer \cite{Schaefer78} who obtained  a famous
dichotomous classification: If $\Gamma$ is Schaefer or $\zerovalid$ or
$\onevalid$, 
then $\SAT(\Gamma)$ is in~$\P$; otherwise $\SAT(\Gamma)$ is $\NP$-complete.

The complexity of the implication problem for $\Gamma$-formul\ae\ was studied
in \cite{SchnoorS08}. The authors obtain a dichotomous classification for $\IMP(\Gamma)$  (i.e., given $\varphi$
and $\psi$ two $\Gamma$-formul\ae, does $\varphi\models \psi$ hold~?):
it is in $\P$ if $\Gamma$ is $\Schaefer$ and $\coNP$-complete otherwise.

Since then and in the recent past, complexity classifications
for many further computational problems for $\Gamma$-formul{\ae} have been
obtained (see \cite{crvo08} for a survey). In particular we will consider 
the following abduction problems.
  \problemdef{$\ABD(\Gamma)$.}
      {$\calA=(\varphi,H,q)$, where $\varphi$ is a $\Gamma$-formula, $H$ is a set of variables, and $q \notin H$ is a variable.}
      {Does there exist $E \subseteq \Lits{H}$ (where $\Lits{H}$ denotes the
set of literals that can be built upon  variables from $H$) such that
$\varphi \land E$ is satisfiable and $\varphi \land E \models q$~?}
  \problemdef{$\ABD[\HP](\Gamma)$.}
      {$\calA=(\varphi,H,q)$, where $\varphi$ is a $\Gamma$-formula, $H$ is a  set of variables, and $q \notin H$ is a variable.}
      {Does there exist $E \subseteq H$ such that $\varphi \land E$ is satisfiable and $\varphi \land E \models q$~?}
Abduction is a nonmonotonic reasoning process, whose  most typical example is  medical diagnosis. Given a knowledge base, here $\varphi$ a formula, a set of variables $H$, the hypotheses, and an observation $q$, we are interested in deciding whether there exists an explanation $E$, i.e., a set of literals built upon $H$ consistent with $\varphi$ such that $\varphi$ and $E$ together entail the observation $q$. The problem $\ABD[\HP]$ refers to positive abduction, where   explanations have to be built upon positive literals only.

According to the classifications obtained in \cite{CrZa06,NoZa08} we will use the
fact that if $\Gamma$ is not $\Schaefer$, then $\ABD(\Gamma)$ is $\SigPtwo$-complete and that
if $\Gamma$ is in addition neither $\zerovalid$, nor $\onevalid$ then $\ABD[\HP](\Gamma)$ is $\SigPtwo$-complete, too.

\smallskip

We want to outline at this point the seeming proximity of argumentation to abduction.
In full propositional logic the abduction problem
and the argumentation problem are equivalent (with respect to polynomial many-one reductions)
since they are both complete for the second level of the polynomial hierarchy (\cite{EiGo95,pawoam03}). Indeed there are very
simple reductions proving this equivalence. We give here exemplary the reductions between $\ABD[\HP]$ and $\ARG$.
\begin{enumerate}\setlength{\itemsep}{\itemseplength}
	\item $\ABD[\HP] \leqlogm \ARG$: $(\varphi, H, q) \mapsto (\Delta, \alpha)$, where\\
				$\Delta := \{\varphi\} \cup H$,
	 			$\alpha := q \wedge \varphi$.
	\item $\ARG \leqlogm \ABD[\HP]$: $(\Delta, \alpha) \mapsto (\varphi,H, q)$, where\\
				$\Delta = \{\varphi_1, \dots, \varphi_n\}$,
				$H := \{x_1, \dots, x_n\}$ where the $x_i's$ are fresh variables,\\
				$\varphi := (\alpha \leftrightarrow q) \land \bigwedge_{i=1}^n (x_i \leftrightarrow \varphi_i)$.
\end{enumerate}

For fragments of propositional logic these reductions do not generally preserve the properties of the chosen fragment and are thus not suited to transfer complete complexity classifications between abduction and argumentation. Nevertheless we will use the idea of the first reduction to transfer certain hardness results from abduction to argumentation.
For instance by the first reduction and hardness results in \cite{NoZa08} one obtains immediately that deciding the existence of a support for $\horn$-formul\ae\ 
is $\NP$-hard. Since for $\horn$-formul\ae\ satisfiability and implication are in $\P$, the verification problem in comparison is in $\P$.

\section{Argumentation problems}\label{sec:argumentation_def}
 In this section we define the computational problems we are interested in. 

\begin{definition} \label{def:argument}{\rm \cite{behu01}}
An {\rm argument} is a pair $\arg{\Phi,\alpha}$, where $\Phi$ is a set of
formul{\ae} and $\alpha$ is a formula such that 
\begin{enumerate}\setlength{\itemsep}{\itemseplength}
\item $\Phi$ is consistent,
\item $\Phi\models\alpha$,
\item $\Phi$ is  minimal with regards to  property (2), \emph{i.e.}, no proper subset
of $\Phi$ entails $\alpha$.
\end{enumerate}
We say that $\arg{\Phi,\alpha}$ is an \emph{argument} for $\alpha$.
If $\Phi\subseteq\Delta$ then it is said to be an  \emph{argument in $\Delta$}. 
We call $\alpha$ the \emph{claim} and $\Phi$ the \emph{support} of the
argument.
\end{definition}

Note that in a more general setting a support $\Phi$ for a claim  $\alpha$ is
a set of formul{\ae} such that  $\Phi$ is consistent and  $\Phi\models\alpha$
and no minimality is required. However, in the definition of an argument, the
support is a minimal one.

\medskip

Let $\Gamma$ be a constraint language.
Then the \emph{argument existence problem for $\Gamma$-formul\ae}\ is defined
as follows:

  \problemdef{$\ARG(\Gamma)$.}
      {$(\Delta,\alpha)$, where $\Delta$ is a set of $\Gamma$-formul\ae\ and $\alpha$ is a $\Gamma$-formula.}
      {Does there exist $\Phi$ such that  $\arg{\Phi,\alpha}$ is an argument in $\Delta$~?}

Besides the decision problem for the existence of an argument we are interested
in the \emph{verification problem} $\ARGCHECK(\Gamma)$ and   in the \emph{relevance problem
$\ARGREL(\Gamma)$, which are defined as follows:}
  \problemdef{$\ARGCHECK(\Gamma)$.}
      {$(\Phi,\alpha)$, where $\Phi$ is a set of $\Gamma$-formul\ae\ and $\alpha$ is a $\Gamma$-formula.}
      {Is $(\Phi,\alpha)$ an argument~?}

\problemdef{$\ARGREL(\Gamma)$.}
      {$(\Delta,\alpha, \psi)$, where $\Delta$ is a set of $\Gamma$-formul\ae, $\psi \in \Delta$ and $\alpha$ is a $\Gamma$-formula.}
      {Does there exist $\Phi$ such that  $\psi \in \Phi$ and $\arg{\Phi,\alpha}$ is an argument in $\Delta$~?}

Let us recall that in the full framework of propositional logic these three problems, $\ARG$, $\ARGCHECK$ and $\ARGREL$, are respectively $\SigPtwo$-complete \cite{pawoam03}, $\DP$-complete and $\SigPtwo$-complete (see e.g. \cite{CrScThWo11}).

\section{Methods and technical tools}\label{sec:tools}

\subsection{Co-clones and Galois connection}\label{subsec:coclone}
We now introduce the logical and algebraic tools that our hardness proofs rely
on. For establishing the complexity of the argumentation problems when
restricted to \mbox{$\Gamma$-formul\ae}, the key will be to study the expressive
power of the set $\Gamma$. This expressivity can be more or less restricted
as discussed in the following definition where the notations from
\cite{SchnoorS08} are adopted. 

\begin{definition}
	Let $\Gamma$ be a constraint language.
\begin{itemize}\setlength{\itemsep}{\itemseplength}
 \item The set $\clos{\Gamma}$ is the smallest set of relations that contains
	$\Gamma$ and the equality constraint, $=$, and which is closed under
primitive
	positive first order definitions, i.e., if
	$\phi$ is a $\Gamma\cup\{=\}$-formula and $R(x_1,\ldots,
x_n)\equiv\exists
	y_1\ldots \exists y_l \phi(x_1,\ldots, x_n, y_1,\ldots, y_l)$, then
	$R\in\clos{\Gamma}$.
	In other words, $\clos{\Gamma}$ is the set of relations that can be
expressed as a $\Gamma\cup\{=\}$-formula with existentially quantified
variables.
\item The set $\closneq{\Gamma}$ is the set of relations that can be
expressed as a $\Gamma$-formula with existentially quantified
variables (no equality relation is allowed).
\item The set $\closnexneq{\Gamma}$ is the set of relations that can be
expressed as a $\Gamma$-formula (neither equality relation nor existential
quantification is allowed).
\end{itemize}
\end{definition}

Let us explain why these closure operators are relevant for us. Assume that
$\Gamma_1\subseteq \closnexneq{\Gamma_2}$. Then any $\Gamma_1$-formula can be
transformed into an equivalent  $\Gamma_2$-formula in replacing every $\Gamma_1$-constraint
by its equivalent $\Gamma_2$-formula. This transformation, which is based on local replacement, 
is computable in logarithmic space (note that both $\Gamma_1$ and $\Gamma_2$ are finite, and not part of the input, so the cost of finding for each relation in 
$\Gamma_1$ an equivalent $\Gamma_2$-formula is not taken into account).
Since for such  equivalent
formulas the answers to the problems that we consider in this paper are the
same, the closure operator $\closnexneq{.}$ directly induces reductions for our
problems, e.g.,  $\ARG(\Gamma_1) \leqlogm \ARG(\Gamma_2).$ 

This notion of expressibility can be relaxed in allowing equality relations and existential
quantification. For some computational problems this is still relevant.
For instance, assume that $\Gamma_1\subseteq \clos{\Gamma_2}$. Then we have a
procedure to transform any $\Gamma_1$-formula into a satisfiability-equivalent $\Gamma_2$-formula:
the equivalent $\Gamma_2$-formula contains additional
existentially quantified first-order variables and equality
constraints can occur. The existential quantifiers can be removed
and the equality constraints can be dealt
with by identification of  variables.
Thus, it has been shown that $\SAT(\Gamma_1)$
can be reduced in logarithmic space to
$\SAT(\Gamma_2)$ (see \cite{Jeavons98,AlBaImScVo05}). Hence, the
complexity of $\SAT(\Gamma)$ depends only on $\clos \Gamma$. 
The set $\clos\Gamma$ is called \emph{relational clone} (or a
\emph{co-clone}). Accordingly, in order to obtain a full
complexity classification for the satisfiability problem one only
has to study the co-clones.

Interestingly, there exists a Galois correspondence between the lattice of
Boolean relations (co-clones) and the lattice of Boolean functions (clones) (see
\cite{Geiger68,bokakoro69}). As a consequence,  based on the famous Post's
description of the lattice of clones \cite{pos41}, the lattice of co-clones is
nowadays well-known (see e.g., \cite{BohlerRSV05,CreignouKZ08}). Therefore, this
Galois connection and this lattice provide a very powerful tool that can be
successfully applied in order to obtain complexity classifications for
computational problems dealing with Boolean formul{\ae} (see e.g.,
\cite{crvo08} for a survey  and \cite{NoZa08} for certain variants of the
abduction problem).

However, 
  this Galois connection  is apparently not appropriate in order to
transfer   complexity results in the case of argumentation.  Indeed, suppose
that $\varphi(\overline{x})$ is logically equivalent to $\exists
\overline{y}\varphi'(\overline{x},\overline{y})$.  It is clear that $\varphi$ is
satisfiable if and only if $\varphi'$ is satisfiable. Moreover,   for any
formula $\psi(\overline{x})$ we have that   $\varphi\models\psi$ if and only if
$\varphi'\models\psi$.  However, if $\psi(\overline{x})$ itself is logically
equivalent to $\exists
\overline{u}\psi'(\overline{x},\overline{u})$, it is not true any more that
$\varphi\models \psi$ implies $\varphi\models \psi'$ (and neither $\varphi'\models \psi'$).  Therefore, when
transforming instances between argumentation problems,
introducing existential variables is problematic with respect to the claim. For
this reason we will
introduce a technical version of the two problems
$\ARGCHECK$ and $\ARGREL$
in which we can differentiate the restrictions put on the knowledge
base from the ones put on the claim. The variants we will use are defined as
follows.

\problemdef{$\ARGCHECK(\Gamma, R)$.}
      {$(\Phi,\alpha)$, where $\Phi$ is a set of $\Gamma$-formul\ae\ and
$\alpha$ is an $R$-constraint.}
      {Is $(\Phi,\alpha)$ an argument?}

\problemdef{$\ARGREL(\Gamma, R)$.}
      {$(\Delta,\alpha, \psi)$, where $\Delta$ is a set of $\Gamma$-formul\ae,
$\psi \in \Delta$ and $\alpha$ is an \mbox{$R$-constraint}.}
      {Does there exist $\Phi \subseteq \Delta$ such that
	      \begin{enumerate}\setlength{\itemsep}{\itemseplength}
	      	\item $\psi \in \Phi$ and
	      	\item $(\Phi,\alpha)$ is an argument?
	      \end{enumerate}
			}

Also, it is not clear how to get rid of the equality
constraints. Indeed
identifying variables that are connected by equality constraints does not
necessarily preserve minimality of the support.

For these two reasons, it is not clear how to prove that the complexity of the argumentation problems only depends on the relational clone $\clos{\Gamma}$.
The best we can obtain is the following key lemma, which will be of use for the  
classifications for $\ARGCHECK$ and $\ARGREL$.

\begin{lemma}\label{lem:not_Schaefer_corr}
	Let $\Gamma, \Gamma'$ be two constraint languages and $R$ a
Boolean relation. If $\Gamma' \subseteq
\closneq{\Gamma}$
and $R\in\closnexneq\Gamma$ then
	\begin{enumerate}\setlength{\itemsep}{\itemseplength}
	\item	$\ARGCHECK(\Gamma', R) \leqlogm \ARGCHECK(\Gamma).$
	\item $\ARGREL(\Gamma', R) \leqlogm \ARGREL(\Gamma).$
	\end{enumerate}
\end{lemma}

\begin{proof}
\begin{enumerate}\setlength{\itemsep}{\itemseplength}
\item
Let $(\Phi,\alpha)$ be an instance of the first problem, where
$\Phi=\{\delta_i\mid i\in I\}$ for some index set $I$ and $\alpha = R(x_1, \dots, x_k)$. We map this
instance to $(\Phi', \alpha')$, where $\Phi' = \{\delta_i' \mid \delta_i \in
\Phi\}$ and $\alpha'$ is a $\Gamma$-formula equivalent to $R(x_1, \dots, x_k)$.
For $i \in I$ we obtain $\delta_i'$ from $\delta_i$ by   replacing $\delta_i$ by
an equivalent $\Gamma$-formula with
existential quantifiers (such a representation exists since $\Gamma' \subseteq
\closneq{\Gamma}$) and deleting all existential quantifiers.

\item
Let $(\Delta,\alpha, \delta_1)$ be an instance of the first problem, where
$\Delta=\{\delta_i\mid i\in I\}$ for some index set $I$ and $\alpha = R(x_1, \dots, x_k)$. We map this
instance to $(\Delta', \alpha', \delta_1')$, where $\Delta' = \{\delta_i' \mid
\delta_i \in \Delta\}$ and $\alpha'$ is a $\Gamma$-formula equivalent to $R(x_1,
\dots, x_k)$.
For $i \in I$ we obtain $\delta_i'$ from $\delta_i$ by the same procedure as in
the previous case.
\end{enumerate}
\end{proof}

As we discussed above the complexity of the verification and the relevance
problem when restricted to
$\Gamma$-formul{\ae} is not \textit{a priori}
  completely determined by the relational clone $\clos\Gamma$. However due to 
the above lemma, the lattice of Boolean co-clones together with the mentioned Galois
connection will still be of help.

\subsection{Some co-clones and various
expressibility lemmas}\label{subsec:technical_lemmas}

In this subsection we recall the relevant knowledge on the lattice of
co-clones and give some technical expressibility results that will be of use
for the proofs.

For the results referring to the lattice of co-clones we use the notations and
the results  from \cite{CreignouKZ08}.

\begin{lemma}\label{lem:CocloneIS}
The smallest co-clone that contains all $\positive$ (resp., $\negative$) relations
is $\CoCloneIS_0$ (resp., $\CoCloneIS_1$).  
A relation $R$ is in $\CoCloneIS_0$ (resp., $\CoCloneIS_1$) if and only if  $m,
m'\in R$ implies $m\rightarrow m'\in R$ (resp.,  $m\not\rightarrow m'\in R$),
where the binary operator $\rightarrow$ (resp., $\not\rightarrow$) applied on
Boolean vectors is applied coordinate-wise.
\end{lemma}

\begin{remark}
 Observe that there are relations in $\CoCloneIS_0$ (resp., $\CoCloneIS_1$)
which are not $\positive$ (resp., $\negative$), for instance the equality
relation.
\end{remark}

\begin{lemma}\label{lem:impl_CocloneInclusions}
Let $\Gamma$ be a constraint language which is not Schaefer. 
\begin{itemize}\setlength{\itemsep}{\itemseplength}
\item 
If $\Gamma$ is not $\complementive$, but is $\zerovalid$ and 
$\onevalid$, then $\clos{\Gamma}$
contains all relations that are both $\zerovalid$ and $\onevalid$.
\item
If $\Gamma$ is not $\complementive$, not $\zerovalid$
but $\onevalid$ (resp. not $\onevalid$ but
$\zerovalid$), then
$\clos{\Gamma}$ contains all relations that are $\onevalid$
($\zerovalid$).
\item 
If $\Gamma$ is not $\complementive$, not $\zerovalid$,
not $\onevalid$, then
$\clos{\Gamma}$ contains all relations.
\end{itemize}
\end{lemma}

Let us now give some expressibility results that we will use in
our hardness proofs. In the proofs of the following lemmas $V=\{x_1, \dots, x_k\}$
will denote a set of $k$ distinct variables. We will suppose w.l.o.g
that the constraint language $\Gamma$ consists of a single relation $R$ of arity $k$.
The reason why we can assume this is that, w.r.t. expressivity, any finite $\Gamma = \{R_1, \dots, R_n\}$
can be 'condensed' to a single relation by the Cartesian product $R = R_1 \times \dots \times R_n$.
It clearly holds that $R \in \closnexneq{\Gamma}$ and $R$ has all properties that $\Gamma$ has.

\begin{lemma}\label{lem:implementations}
Let $\Gamma$  be a constraint language. If $\Gamma$ is
\begin{enumerate}\setlength{\itemsep}{\itemseplength}
\item $\complementive$, but neither $\onevalid$ nor
$\zerovalid$,	then $(x\neq y)\in\closnexneq\Gamma$;

\item  not $\complementive$, but $\onevalid$ and
$\zerovalid$,	then $(x \rightarrow y)\in\closnexneq\Gamma$;

\item neither $\complementive$ nor $\onevalid$, nor
$\zerovalid$, then $(x \land \neg{y})\in\closnexneq\Gamma$;

\item $\onevalid$ and not $\zerovalid$,
then  $\T\in\closnexneq\Gamma$.
\end{enumerate}
\end{lemma}

\begin{proof}
Folklore, see e.g., \cite{crkhsu01}.
\end{proof}

\begin{lemma}\label{lem:impl_equality}
Let $\Gamma$ be a constraint language that is both $\zerovalid$ and $\onevalid$.
Then $(x = y)\in\closnexneq\Gamma$.
\end{lemma}
\begin{proof}
Let $R\in \Gamma$ be a $k$-ary relation. Since $R\ne\{0,1\}^k$ there is an 
$m \notin R$ and $m \neq 0^k$ and $m \neq 1^k$.
For $i\in\{0,1\}$, set $V_i=\{x\mid x\in V, m(x)=i\}$.
We observe that the sets	$V_0$ and $V_1$ are nonempty (since $m \neq 0^k$ and $m \neq 1^k$).
Denote by $C$ the $R$-constraint $C=R(x_1,\ldots , x_k)$.
Set $M(x,y)=C[V_0/x,\, V_1/y]$.
It contains $\{00,11\}$ (since $R$ is $\zerovalid$ and $\onevalid$) but not $01$ (since $m \notin R$).
Finally, we have $M(x,y) \land M(y,x) \equiv (x = y)$.
\end{proof}

\begin{lemma}\label{lem:impl_equality_with_T}
Let $\Gamma$ be a constraint language. If $\Gamma$ is:
\begin{enumerate}\setlength{\itemsep}{\itemseplength}
	\item $\onevalid$ but neither $\zerovalid$ nor $\positive$, then  $(x = y) \land z\in\closnexneq\Gamma$;
	\item $\zerovalid$ but neither $\onevalid$ nor $\negative$, then  $(x = y) \land \neg z\in\closnexneq\Gamma$.
\end{enumerate}
\end{lemma}

\begin{proof}
We only prove the first case, the second case can be treated analogously /
dually.
Let w.l.o.g. $\Gamma = \{R\}$, thus $R$ is $\onevalid$ but neither $\zerovalid$ nor
$\positive$. We perform a case distinction according to whether $R\in
\CoCloneIS_0$ or not.

Let us first suppose that $R\in
\CoCloneIS_0$.
 According to \cite{CreignouKZ08}
the relation $R$ ($\in \CoCloneIS_0$) can be written as
a conjunction of positive clauses and equalities. 
If $R$ can be written with no equality, then $R$ is positive, a contradiction.
So, any representation
of $R$ as a conjunction of positive clauses and equalities requires at least one
equality. Suppose thus w.l.o.g that
$R(x_1,\ldots, x_k)\models (x_1=x_2) $,
while   $ R(x_1,\ldots, x_k)\not\models x_1$,
which means that the equality $x_1=x_2$ can be transitively deduced from the
equality constraints occurring in any representation of $R$ (note that such a configuration necessarily occurs,
 otherwise  no equality constraints would be needed, it would be sufficient to write $(x_1)\wedge(x_2)$, contradicting the fact that $R$ is not positive). 
Let $W:=\{x_i\mid
 R(x_1,\ldots, x_k)\models (x_1=x_i)\}$. Observe
that $W'=V\setminus (W\cup\{x_1\})$ is nonempty for $R$ is not
$\zerovalid$.
Denote by $C$ the $R$-constraint $C=R(x_1,\ldots , x_k)$.
Consider the constraint $M(x_1, x_2, y)= C[W/x_2, W'/y]$.  
 One verifies that
$M(x_1, x_2, y)\equiv (x_1=x_2)\land y$.
Therefore, 
$(x=y)\land z\in\closnexneq\Gamma$.

Let  us now suppose that $R\not\in
\CoCloneIS_0$. According to Lemma~\ref{lem:CocloneIS} there are $m_1, m_2 \in R$
such that $m_1 \rightarrow m_2 \notin R$.
For $i,j\in\{0,1\}$, set $V_{i,j}=\{x\mid x\in V, m_1(x)=i\land m_2(x)=j\}$.
Observe that the sets	$V_{0,0}$ and $V_{1,0}$ are nonempty
(otherwise $m_2 = m_1 \rightarrow m_2$ or $1^k = m_1 \rightarrow m_2$, a contradiction).
Denote by $C$ the $R$-constraint $C=R(x_1,\ldots , x_k)$.
Set $M(x_1,x_2,x_3,x_4)=C[V_{0,0}/x_1,\, V_{1,0}/x_2,\, V_{0,1}/x_3,\, V_{1,1}/x_4]$.
It contains $\{1111, 0101, 0011\}$ (since resp. $R$ is $\onevalid$, $m_1 \in R$, $m_2 \in R$) but not $1011$ (since $m_1 \rightarrow m_2 \notin R$).
We conclude that $M(x,y,z,z) \land \T(z)$ contains $\{111, 001\}$ but not $101$.
Finally, we verify that $M(x,y,z,z) \land M(y,x,z,z) \land \T(z) \equiv (x = y) \land z$.
Since by Lemma~\ref{lem:implementations}, $\T\in\closnexneq\Gamma$, we obtain
$(x=y)\land z\in\closnexneq\Gamma$.
\end{proof}

\begin{lemma}\label{lem:impl_not_schaefer_equality}  
Let $\Gamma$ be a constraint language. If $\Gamma$ is not $\Schaefer$, then
  $(x = y)  \in\closneq\Gamma$. In particular, for any relation $R$, if
$R\in\clos\Gamma$ and  $\Gamma$ is not $\Schaefer$, then $R\in\closneq\Gamma$.
\end{lemma}

\begin{proof}
We perform a case distinction according to whether $\Gamma$ is 0/1-valid or not.

If $\Gamma$ is $\zerovalid$ and $\onevalid$, then according to
Lemma~\ref{lem:impl_equality} there is a $\Gamma$-formula equivalent to $(x =
y)$.

If $\Gamma$ is not $\zerovalid$ but $\onevalid$ (resp. not $\onevalid$ but $\zerovalid$), then according to
Lemma~\ref{lem:impl_equality_with_T} there is a $\Gamma$-formula $\varphi(x,y,z)$ equivalent to
$(x = y) \land z$ (resp. $(x = y) \land \neg z$).
Hence, $\exists z\; \varphi(x,y,z)$ fulfills our needs.

At last let $\Gamma$ be neither $\zerovalid$ nor $\onevalid$. It suffices here
to show that we are able to express disequality, $(x \neq y)$, since
$(x = y) \equiv \exists z (x \neq z) \land (z \neq y)$.
If $\Gamma$ is $\complementive$ we conclude by
Lemma~\ref{lem:implementations}, first item.
Therefore, suppose now that $\Gamma$ is not
$\complementive$.
Let w.l.o.g. $\Gamma = \{R\}$.
Since $R$ is not $\horn$, there are $m_1, m_2 \in R$ such that $m_1 \land m_2 \notin R$.
For $i,j\in\{0,1\}$, set $V_{i,j}=\{x\mid x\in V, m_1(x)=i\land m_2(x)=j\}$.
Observe that the sets	$V_{0,1}$ and $V_{1,0}$ are nonempty
(otherwise $m_2 = m_1 \land m_2$ or  $m_1 = m_1 \land m_2$, a contradiction).
Denote by $C$ the $R$-constraint $C=R(x_1,\ldots , x_k)$.
Set $M_1(u,x,y,v)=C[V_{0,0}/u,\, V_{0,1}/x,\, V_{1,0}/y,\, V_{1,1}/v]$.
It contains $\{0011, 0101\}$ (since $m_1,m_2 \in R$)
but it does not contain $0001$ (since $m_1 \land m_2 \notin R$).
Further, since $R$ is not $\dualhorn$, there are $m_3, m_4 \in R$ such that $m_3 \lor m_4 \notin R$.
For $i,j\in\{0,1\}$, set $V'_{i,j}=\{x\mid x\in V,\; m_3(x)=i\, \land\, m_4(x)=j\}$.
Observe that the sets	$V'_{0,1}$ and $V'_{1,0}$ are nonempty.
Set $M_2(u,x,y,v)=C[V'_{0,0}/u,\, V'_{0,1}/x,\, V'_{1,0}/y,\, V'_{1,1}/v]$.
It contains $\{0011, 0101\}$ (since $m_3,m_4 \in R$)
but it does not contain $0111$ (since $m_3 \lor m_4 \notin R$).
Finally consider the $\{R, (t \land \neg f)\}$-formula
$$M(x,y,f,t) = M_1(f,x,y,t) \land M_2(f,x,y,t) \land (t
\land \neg f).$$
One verifies that it is equivalent to $(x \neq y) \land
(t \land \neg f)$.
Due to the third item of Lemma~\ref{lem:implementations},
$(t \land \neg f)$ is expressible as a $\Gamma$-formula, and therefore so is
$M(x,y,f,t)$.
We conclude observing that $\exists t,f\; M(x,y,f,t)$ is
equivalent to $(x \neq y)$.
\end{proof}

\section{The complexity of $\ARG$}\label{sec:complexity_arg}
 The complexity of deciding the existence of an argument rests on two sources:
finding a candidate support, and checking that it  is consistent
and proves $\alpha$. Observe that the minimality condition plays no role
here: there exists a minimal support if and only if there exists a support.
Therefore, the problem $\ARG(\Gamma)$ lies in the class $\SigPtwo$. When there
is a natural candidate as a support, then the complexity of $\ARG(\Gamma)$
drops to the class
$\coNP$, whereas when satisfiability and implication are tractable then
the complexity drops to the class $\NP$.

\begin{proposition}\label{prop:schaefer_notvalid}
Let $\Gamma$ be a constraint language which is $\Schaefer$, but  neither
$\onevalid$, nor $\zerovalid$.
Then $\ARG(\Gamma)$ is $\NP$-complete.
\end{proposition}

\begin{proof}
The $\NP$-membership follows from the fact that since $\Gamma$ is $\Schaefer$ 
$\SAT(\Gamma)$ and $\IMP(\Gamma)$ are in $\P$ and thus a guessed argument can be
verified in polynomial time.
For the hardness proof we perform a case distinction according to whether
$\Gamma$
is $\complementive$ or not. Suppose first that every relation 
in $\Gamma$ is $\complementive$. We prove the following sequence of reductions: 
$$\ThreeSAT \leqlogm \ARG(\{x\neq y\}) \leqlogm \ARG(\Gamma).$$
The last reduction holds by Item 1  in  Lemma~\ref{lem:implementations}.
For the first 
 reduction let $\varphi = \bigwedge_{i=1}^k C_i$ be an instance
of $\ThreeSAT$  where $\var{(\varphi)} = \{x_1, \dots, x_n\}$.
Let $c_1, \dots, c_k$, $x'_1, \dots, x'_n$, $f$ be fresh variables.
We map $\varphi$ to $(\Delta, \alpha)$ where
$$\begin{array}{rl}
\Delta   =\; &  \bigcup_{j=1}^n\{x_j \neq f, x'_j \neq f\}\\
			\cup\; &  \{\bigwedge_{j=1}^n (x_j \neq x'_j)\}\;\\ 
		  \cup\; &  \bigcup_{i=1,\ldots, k, j=1,\ldots, n}\{x_j \neq
c_i \mid \neg x_j \in C_i\}\cup \{x'_j \neq c_i \mid x_j \in C_i\}, \\
\alpha   =\; &  \bigwedge_{i=1}^k (c_i \neq f) \land \bigwedge_{j=1}^n (x_j \neq x'_j).
\end{array}$$
One can check that $\varphi$ is satisfiable if and only if there exists a $\Phi\subseteq\Delta$ such that
$(\Phi,\alpha)$ is an argument. Intuitively, $x_j'$ plays the role of $\neg
x_j$, 
for every $j$ at most one of the constraints  $x_j \neq f$ and 
$x_j' \neq f$  can appear in the support of an argument, thus
allowing  to identify true literals, while for every $i$  the
constraints  $x_j \neq c_i$ and $x_j' \neq c_i$  are used to certify
that the clause $C_i$ is satisfied.

Second, let us suppose that $\Gamma$ is not $\complementive$. We prove the
following:
$$\PosOneThreeSAT \leqlogm \ARG( \{x \land \neg{y}\}) \leqlogm \ARG(\Gamma).$$
The last reduction follows by  Item 3  in 
Lemma~\ref{lem:implementations}.
For the first one we start  from the $\NP$-complete problem $\PosOneThreeSAT$ in which the
instance is a set of positive 3-clauses and the question is to decide whether
there exists a truth assignment such that each clause contains exactly one true
variable.  
Let $\varphi = \bigwedge_{i=1}^k (x_i \lor y_i \lor z_i)$ be an instance of the
first problem and let
$c_1, \dots, c_k$, $f$ be fresh variables.
We map $\varphi$ to $(\Delta, \alpha)$ where
$$\begin{array}{rl}
\Delta  =\; & \bigcup_{i=1}^k\{ c_i \land x_i \land \neg{y_i} \land \neg{z_i}  \land \neg{f}\} \\
     \cup\;	& \bigcup_{i=1}^k\{ c_i \land \neg{x_i} \land y_i \land \neg{z_i} \land \neg{f}\} \\
		 \cup\; & \bigcup_{i=1}^k\{ c_i \land \neg{x_i} \land \neg{y_i} \land z_i \land \neg{f}\}, \\
\alpha =\;& (c_1 \land \neg{f}) \land \dots \land (c_k \land \neg{f}).
\end{array}$$
Observe that every formula in $\Delta$ can be written as  a $\{x \land \neg
y\}$-formula.
One can check that there is a truth assignment such that each clause $C_i$ contains exactly one variable set to true if and only if
$(\Delta, \alpha)$ admits an argument. Observe that for every $i$ such an argument contains exactly one of the three formul\ae\ involving
$c_i$, thus providing a desired satisfying assignment.
\end{proof}

\begin{proposition}\label{prop:notschaefer_notvalid}
Let $\Gamma$ be a constraint language which is neither $\Schaefer$, nor
$\onevalid$, nor $\zerovalid$.
Then $\ARG(\Gamma)$ is $\SigPtwo$-complete.
\end{proposition}

\begin{proof}
We give a reduction from $\ABD[\HP](\Gamma)$ which is $\SigPtwo$-complete according to \cite{NoZa08}.
We perform a case distinction according to whether $\Gamma$ is $\complementive$
or not.

Suppose first that every relation in $\Gamma$ is $\complementive$.
We show:
$$\ABD[\HP](\Gamma) \leqlogm \ARG(\Gamma \cup \{x \neq y\}) \leqlogm \ARG(\Gamma).$$
The last reduction follows by Item 1  in   Lemma~\ref{lem:implementations}.
For the first one
we map $(\varphi, H, q)$, an instance of $\ABD(\Gamma)$, to $(\Delta,
\alpha)$, where we introduce a fresh variable $f$ and define
%
$$\Delta =\; \{\varphi\}\; \cup\; \{(h \neq f) \mid h \in H\},\ \ \ \ \ \alpha =\; (q \neq f).$$
The proof that the reduction is correct relies on the fact that all formul{\ae} 
occurring in the so obtained instance are $\complementive$, i.e., it suffices to observe correctness for $(\Delta[f/0], \alpha[f/0])$.

In the case where  $\Gamma$ is not $\complementive$ we show
$$\ABD[\HP](\Gamma) \leqlogm \ARG(\Gamma \cup \{x \land \neg{y}\}) \leqlogm \ARG(\Gamma).$$
The last reduction follows by   Item 3  in  Lemma
Lemma~\ref{lem:implementations}. For the first one 
we map $(\varphi, H, q)$, an instance of the first problem, to $(\Delta,
\alpha)$,
where we introduce two fresh variables $t, f$ and define 
%
$$\Delta =\; \{\varphi\}\; \cup\; \{h\land \neg f\mid h\in H\}\; \cup\; \{t\land \neg f\},\ \ \ \ \ \alpha =\; (q \land \neg{f})\land  (t \land \neg{f}).$$
Observe that $\Delta$ is made of $\Gamma$- and $\{x \land \neg{y}\}$-formul{\ae}.
It is easy to check that  $(\varphi, H, q)$ is a
positive instance of the abduction problem if and only if there exists a
support for $\alpha$ in $\Delta$.
\end{proof}
We are now in a position to state the classification theorem.

\begin{theorem}
\label{thm:arg}
 Let $\Gamma$ be a constraint language.
The decision problem $\ARG(\Gamma)$ is
\begin{enumerate}\setlength{\itemsep}{\itemseplength}
 \item in  $\P$ if $\Gamma$ is $\Schaefer$ and $\varepsilon$-valid,
 \item $\NP$-complete if $\Gamma$ is $\Schaefer$ and not $\varepsilon$-valid,
  \item $\coNP$-complete if $\Gamma$ is not $\Schaefer$ and is $\varepsilon$-valid,
 \item $\SigPtwo$-complete if $\Gamma$ is not $\Schaefer$ and not
$\varepsilon$-valid.
\end{enumerate}
\end{theorem}

\begin{proof}
 \begin{enumerate}\setlength{\itemsep}{\itemseplength}
  \item One easily observes that, due to the fact that $\Gamma$ is $\onevalid$
or $\zerovalid$, an instance $(\Delta, \alpha)$ of $\ARG(\Gamma)$
has a solution if and only if $\Delta$ implies $\alpha$. This condition can be
checked in polynomial time
since $\Gamma$ is $\Schaefer$ and thus $\IMP(\Gamma)$ is in $\P$.
  \item Follows from Proposition \ref{prop:schaefer_notvalid}.
  \item One easily observes that, due to the fact that $\Gamma$ is $\onevalid$
or $\zerovalid$, an instance $(\Delta, \alpha)$ of $\ARG(\Gamma)$
has a solution if and only if $\Delta$ implies $\alpha$. This condition can be
checked in $\coNP$ since
$\IMP(\Gamma)$ is in $\coNP$.

To prove $\coNP$-hardness we give a reduction from the $\coNP$-complete problem
$\IMP(\Gamma)$. We map $(\varphi, \psi)$ an instance of the first
problem to $(\{\varphi\}, \psi)$.
  \item Follows from Proposition \ref{prop:notschaefer_notvalid}.
 \end{enumerate}
\end{proof}

\section{The complexity of $\ARGCHECK$}\label{sec:complexity_argcheck}
 In this section we give the complexity classification for the verification
problem. As discussed in  Section \ref{subsec:coclone} the Galois connection
  does not hold\textit{ a
priori} for this problem. However, the following theorem shows that it holds
\textit{a
posteriori}, this means that the dichotomy follows the borders of Post's
lattice, i.e., the complexity of $\ARGCHECK(\Gamma)$ depends on the relational clone
$\clos{\Gamma}$ only. 

\begin{theorem}\label{thm:argcheck}
Let $\Gamma$ be a constraint language.
Then the decision problem $\ARGCHECK(\Gamma)$ is
\begin{enumerate}\setlength{\itemsep}{\itemseplength}
 \item in $\P$ if $\Gamma$ is $\Schaefer$,
 \item $\DP$-complete if $\Gamma$ is not $\Schaefer$.
\end{enumerate}
\end{theorem}


The argument verification problem is in $\DP$. Indeed  $\ARGCHECK = A \cap B$, with $A=\{ (\Delta,\Phi,\alpha) \mid 
      	\Phi \text{ is satisfiable}, \forall \varphi \in \Phi :
 			\Phi\setminus\{\varphi\} \not \models \alpha\}$ and $B=\{ (\Delta,\Phi,\alpha) \mid \Phi \models \alpha\}$, and
  $A\in \NP$ and $B\in \coNP$.

\begin{proposition}\label{prop:tractable_argcheck}
Let $\Gamma$ be a constraint language that is $\Schaefer$. Then $\ARGCHECK(\Gamma)$ is in $\P$.
\end{proposition}

\begin{proof}
Use that $\SAT(\Gamma)$ and $\IMP(\Gamma)$ are in $\P$.
\end{proof}

\begin{proposition}\label{prop:hard_argcheck_noncp}
Let $\Gamma$ be a constraint language which is neither   $\Schaefer$ 
nor   $\complementive$.
Then $\ARGCHECK(\Gamma)$ is $\DP$-complete.
\end{proposition}

\begin{proof}
For the hardness we give a reduction from $\CRITICALSAT$, a $\DP$-complete
problem according to \cite{pawo88}.
We will use  as an
intermediate problem the variant of $\ARGCHECK$ with two parameters,
$\ARGCHECK(\Gamma', R)$ as defined in Section \ref{subsec:coclone},
differentiating the restrictions put on the knowledge base from the ones put on
the claim.

We  perform a case distinction according to whether $\Gamma$ is $\zerovalid$
and/or $\onevalid$.
Throughout the proof we denote by $\varphi = \bigwedge_{i=1}^k C_i$ an instance
of $\CRITICALSAT$.

Suppose first that $\Gamma$ is both $\zerovalid$ and $\onevalid$. We prove
for some well-chosen constraint language $\Gamma' \subseteq \clos{\Gamma}$ the following sequence of reductions:
$$\begin{array}{rl}
	\CRITICALSAT & \leqlogm \ARGCHECK(\Gamma',x \rightarrow y) \\
	             & \leqlogm \ARGCHECK(\Gamma).
\end{array}$$
For the first reduction we associate with $\varphi$ the instance
$(\Phi, \alpha)$ where
$$\begin{array}{rl}

	\Phi   & = \{C_i \lor (f \rightarrow t) \mid i=1,\ldots , k\}, \\
	\alpha & = (f \rightarrow t),
\end{array}$$
with $f,t$ fresh variables.
It is easy to see that $\varphi$ is a critical instance if and only if $(\Phi, \alpha)$ is an argument.

For the second reduction
observe that all formul\ae\ in $\Phi$ are constraints built upon a finite set $\Gamma'$ of
relations  which are $\onevalid$ and $\zerovalid$ and thus $\Gamma' \subseteq
\clos{\Gamma}$ according to Lemma \ref{lem:impl_CocloneInclusions}. Since
$\Gamma$ is not Schaefer, following Lemma \ref{lem:impl_not_schaefer_equality}
we have $\Gamma' \subseteq
\closneq{\Gamma}$. Further,
the relation $x \rightarrow y$ can be expressed by a $\Gamma$-formula according
to Item 2 of Lemma~\ref{lem:implementations}. With this, the second reduction
follows by
Lemma~\ref{lem:not_Schaefer_corr}.

\medskip

Suppose now that 
$\Gamma$ is $\onevalid$ and not $\zerovalid$.
The other case ($\zerovalid$ and not $\onevalid$) can be treated analogously.
We show for some well-chosen constraint language $\Gamma' \subseteq \clos{\Gamma}$ that
$$\begin{array}{rl}
	\CRITICALSAT & \leqlogm \ARGCHECK(\Gamma', \T) \\
	             & \leqlogm \ARGCHECK(\Gamma).
\end{array}$$
For the first reduction we associate with $\varphi$ the instance
$(\Phi, \alpha)$ where
$$\begin{array}{rl}
	\Phi   & = \{C_i \lor u \mid i=1,\ldots , k\}, \\
	\alpha & =  u,
\end{array}$$
with $u$ being a fresh variable.
It is easy to see that $\varphi$ is a critical instance if and only if $(\Phi, \alpha)$ is an argument.

For the second reduction
observe that all formul\ae\ in $\Phi$ are constraints built upon a finite set $\Gamma'$ of
relations  which are $\onevalid$ and thus $\Gamma' \subseteq \clos{\Gamma}$
according to Lemma \ref{lem:impl_CocloneInclusions}. Since
$\Gamma$ is not Schaefer, following Lemma \ref{lem:impl_not_schaefer_equality}
we have $\Gamma' \subseteq
\closneq{\Gamma}$. Further,
the relation $\T$ can be expressed by a $\Gamma$-formula according to Lemma~\ref{lem:implementations}. With this, the second reduction follows by
Lemma~\ref{lem:not_Schaefer_corr}.

\medskip

Finally suppose that $\Gamma$ is neither $\onevalid$ nor $\zerovalid$.
We show for some well-chosen constraint language $\Gamma' \subseteq \clos{\Gamma}$ that
$$\begin{array}{rl}
	\CRITICALSAT & \leqlogm \ARGCHECK(\Gamma', x \land \neg y) \\
	             & \leqlogm \ARGCHECK(\Gamma).
\end{array}$$
For the first reduction we associate with $\varphi$ the instance
$(\Phi, \alpha)$ where
$$\begin{array}{rl}
	\Phi   & = \{(C_i \lor u) \land \neg v  \mid i=1,\ldots , k\}, \\
	\alpha & =  u \land \neg v,
\end{array}$$
with $u,v$ fresh variables.
It is easy to see that $\varphi$ is a critical instance if and only if $(\Phi, \alpha)$ is an argument.

For the second reduction
observe that all formul\ae\ in $\Phi$ are constraints built upon a finite set $\Gamma'$ of
relations and thus $\Gamma' \subseteq \clos{\Gamma}$ according to Lemma
\ref{lem:impl_CocloneInclusions}. Since
$\Gamma$ is not Schaefer, following Lemma \ref{lem:impl_not_schaefer_equality}
we have $\Gamma' \subseteq
\closneq{\Gamma}$. Further,
the relation $x \land \neg y$ can be expressed by a $\Gamma$-formula according
to Item 3 of  Lemma~\ref{lem:implementations}. With this, the second reduction
follows by
Lemma~\ref{lem:not_Schaefer_corr}.
\end{proof}


\begin{proposition}\label{prop:hard_argcheck_cp}
Let $\Gamma$ be a constraint language which is not  $\Schaefer$ 
but is   $\complementive$.
Then $\ARGCHECK(\Gamma)$ is $\DP$-complete.
\end{proposition}

\begin{proof}
We prove that 
$\ARGCHECK(\Gamma \cup \{\T\}) \leqlogm \ARGCHECK(\Gamma).$
This will prove hardness for $\ARGCHECK(\Gamma)$ since $\Gamma \cup \{T\}$ is
neither $\Schaefer$ nor $\complementive$ (because of $\T$) and therefore $\ARGCHECK(\Gamma \cup \{T\})$ is a
$\DP$-complete problem according to Proposition~\ref{prop:hard_argcheck_noncp}.

Suppose first that $\Gamma$ is both $\zerovalid$ and $\onevalid$. Then we show
that $ \ARGCHECK(\Gamma \cup \{\T\})  \leqlogm
\ARGCHECK(\Gamma \cup\{=\})\leqlogm
\ARGCHECK(\Gamma).$ The second reduction  holds according to Lemma
\ref{lem:impl_equality}. For the first one 
let $(\Phi,\alpha)$ be an instance of $\ARGCHECK(\Gamma \cup \{\T\})$. 
Introduce a fresh variable $t$ and replace in all formul\ae\ all
$\T$-constraints $\T(x)$ by $(x = t)$. Thus we obtain $(\Phi',\alpha')$ an instance of
$\ARGCHECK(\Gamma\cup\{=\})$. The key to observe that this reduction is correct
is that $\Gamma$ is $\complementive$, i.e., it suffices to observe correctness for $(\Phi'[t/1], \alpha[t/1])$.

In the case $\Gamma$ is neither $\zerovalid$ nor $\onevalid$, then we show
that $ \ARGCHECK(\Gamma \cup \{\T\})  \leqlogm
\ARGCHECK(\Gamma \cup\{\ne\})\leqlogm
\ARGCHECK(\Gamma).$ The second reduction  holds according to Item 1 of Lemma
\ref{lem:implementations}. For the first one 
let $(\Phi,\alpha)$ be an instance of $\ARGCHECK(\Gamma \cup \{\T\})$. 
Introduce a fresh variable $t$ and introduce in all formul\ae\ for each 
$\T$-constraint $\T(x)$ a new variable $f_x$ and replace $\T(x)$ by the two disequality constraints $(x\ne f_x)$ and $(f_x\ne t)$.
Thus we obtain $(\Phi',\alpha')$ an instance of
$\ARGCHECK(\Gamma\cup\{\neq\})$. Again, the key to observe correctness is that 
one may restrict attention to the case $t = 1$.
\end{proof}

\section{The complexity of $\ARGREL$}\label{sec:complexity_argrel}

In this section we give the complexity classification for the relevance
problem. As for the verification problem the  Galois connection does not hold
\textit{a
priori} for this problem (see the discussion in Section \ref{subsec:coclone}).
However,  interestingly and contrary to the verification problem, it  does not
hold \textit{a posteriori} either (unless $\P = \NP$). We reveal constraint languages
$\Gamma_1$ and $\Gamma_2$ such that $\clos{\Gamma_1}=\clos{\Gamma_2}$, and 
$\ARGREL(\Gamma_1)$ is in $\P$ whereas $\ARGREL(\Gamma_2)$ is $\NP$-complete.
As we will see, it is the equality relation, $=$, that is responsible for the
discrepancy in complexity.

\begin{theorem}\label{thm:arg_rel}
Let $\Gamma$ be a constraint language. Then the decision problem
$\ARGREL(\Gamma)$ is
\begin{enumerate}\setlength{\itemsep}{\itemseplength}
 \item in $\P$ if $\Gamma$ is positive or negative,
 \item $\NP$-complete if $\Gamma$ is $\Schaefer$ but neither positive, nor
negative,
 \item $\SigPtwo$-complete if $\Gamma$ is not $\Schaefer$.
\end{enumerate}
\end{theorem}

\begin{remark}
 Consider
 the two relations $R(x,y)=(x\lor y)$ and $R'(x,y,z)=(x\lor y)\land (y=z)$. 
Observe that these two relations generate the same relational clone, $\clos{\{R\}}=\clos{\{R'\}}$, since $R'$ can be expressed as the conjunction of an $R$-contraint and an equality constraint, and $R(x,y)=\exists z R'(x,y,z)$. However, the relation $R$ is positive, and the relation $R'$ is Schaefer (bijunctive) but not positive. Therefore, according to the previous theorem $\ARGREL(\{R\})$ is in $\P$ whereas $\ARGREL(\{R'\})$ is $\NP$-complete.
\end{remark}

\begin{remark}\label{rem:useful_remark}
	 Observe that we inherit all hardness results from $\ARG(\Gamma)$ via
the reduction $(\Delta, \alpha) \mapsto (\Delta \cup \{\phi\}, \alpha \land
\phi, \phi)$, where $\phi$ is any non-trivial $\Gamma$-formula
made of fresh variables. 
\end{remark}

\medskip

While the polynomial cases of $\ARG$ and $\ARGCHECK$ relied directly on the
tractability of $\SAT$ and $\IMP$, for $\ARGREL$ we need to investigate the
structure of the problem.

\begin{proposition}\label{prop:prop_argrel_0}
Let $\Gamma$ be a constraint language. If $\Gamma$ is either $\positive$ or
$\negative$, then $\ARGREL(\Gamma)$ is in $\P$.
\end{proposition}
\begin{proof}
We treat only the case of $\positive$ $\Gamma$, the other case can be treated analogously / dually.
In this case $\alpha$ and the formul\ae\ in the knowledge base $\Delta$ can be
considered as positive CNF-formul\ae.
We claim that Algorithm~1 decides $\ARGREL(\Gamma)$ in polynomial time.
\begin{algorithm}[ht]
\begin{algorithmic}
    \REQUIRE a set $\Delta$ of positive formul\ae, $\psi \in \Delta$ and a positive formula $\alpha = \bigwedge_{i\in I}C_i$.
      \FORALL{$i \in I$}
        \STATE $\Delta_i := \{\psi\}\  \cup\  \{\delta \in \Delta \mid \delta \not\models C_i\}$
        \IF{$\Delta_i \models \alpha$} \label{alg:condition}
          \STATE accept
        \ENDIF
      \ENDFOR
      \STATE reject
\end{algorithmic}
  \caption{Algorithm for $\ARGREL(\Gamma)$ for $\positive$ $\Gamma$.}
  \end{algorithm}
 	The running time of Algorithm~1 is obviously polynomial (the test $\Delta_i \models \alpha$ is an instance of $\IMP(\Gamma)$, which is in $\P$ for $\positive$ $\Gamma$).

  To prove correctness, we need the following easy but crucial observation.
  \begin{observation}\label{obs:star}
	Let $a,b$ be positive CNF-formul\ae\ and let $\gamma$ be a positive
clause. If $a \not\models \gamma$ and $b \not\models \gamma$, then $a \land b
\not\models \gamma$.
	\end{observation}
  If Algorithm~1 accepts, then there exists a $\Delta_i \subseteq \Delta$ such that $\Delta_i \models \alpha$
  and no $\delta \in \Delta_i\setminus \{\psi\}$ entails $C_i$. With
Observation~\ref{obs:star} we obtain that $\Delta_i\setminus \{\psi\} \not
\models C_i$, therefore
  $\Delta_i\setminus \{\psi\} \not \models \alpha$.
  We conclude that $\Delta_i$ contains a minimal support $\Phi$ such that $\psi \in \Phi$.
  
  Conversely, let $\Phi$ be a minimal support such that $\psi \in \Phi$. Since $\Phi\setminus \{\psi\} \not \models \alpha$, there
  is at least one $i$ such that $\Phi\setminus \{\psi\} \not \models C_i$, i.e.,
in particular no $\delta \in \Phi\setminus \{\psi\}$ entails $C_i$.
  For this $i$ the algorithm constructs $\Delta_i := \{\psi\}\  \cup\  \{\delta \in \Delta \mid \delta \not\models C_i\}$.
  Obviously $\Phi \subseteq \Delta_i$, and since $\Phi \models \alpha$ we obtain that $\Delta_i \models \alpha$ which causes the algorithm to accept.

Note that the same algorithmic idea was applied in \cite[Proposition~3.8]{CrScThWo11} to solve the relevance problem for positive terms.
\end{proof}

Let us now turn to the $\NP$-complete case, when $\Gamma$ is Schaefer but
neither positive nor negative.
Observe that if $\Gamma$ is Schaefer, then $\ARGCHECK(\Gamma)$ is in $\P$ (see
Theorem \ref{thm:argcheck}), and therefore $\ARGREL(\Gamma)$ is in $\NP$: 
Guess a $\Phi$ and verify that $\psi \in \Phi$ and $(\Phi, \alpha) \in
\ARGCHECK$.
The hardness proofs rely on the following basic hardness results.

\begin{lemma}\label{lem:hardness_for_eq_argrel}
 $\ARGREL(\{x=y\})$,  $\ARGREL(\{(x=y)\land z\})$ and  $\ARGREL(\{(x=y)\land \neg z\})$ are $\NP$-hard.
\end{lemma}
\begin{proof}
For $\ARGREL(\{x=y\})$ we give a reduction from $\ThreeSAT$.
Let $\varphi = \bigwedge_{i=1}^k C_i$, $\Vars{\varphi} = \{x_1, \dots, x_n\}$.
Let $c_0, c_1, \dots c_k, s$ be fresh variables. We map $\varphi$ to the
instance
$(\Delta, \alpha, \psi)$ defined as follows.
$$\begin{array}{rl}
\Delta  =\; & \{\gamma_j, \delta_j \mid 1 \leq j \leq n\}\; \cup\;  \{\psi\} \\
\gamma_j =\; & (c_0 = x_j) \land \bigwedge_{i \text{ s.t. } x_j \in C_i}
(c_{i-1} = c_i) \\
\delta_j =\; & (x_j = s) \land \bigwedge_{i \text{ s.t. }\neg x_j \in C_i}
(c_{i-1} = c_i) \\
\alpha  =\; & (c_0 =s) \\
\psi =\; & (c_k = s) \\
\end{array}$$
Correctness is not difficult to observe.
There is a one-to-one correspondence between (not necessarily minimal) supports
$\Phi$ in which $\psi$ is relevant and satisfying assignments $\sigma$ for
$\varphi$
given by
$\gamma_j \in \Phi$ iff $\sigma(x_j) = 1$ and
$\delta_j \in \Phi$ iff $\sigma(x_j) = 0$.
A support (containing a relevant $\psi$) does never contain both $\gamma_j$ and
$\delta_j$, since otherwise $\psi$ would not be relevant.

For $\ARGREL(\{(x=y)\land z\})$ (resp. $\ARGREL(\{(x=y)\land \neg z\})$) we use the same reduction scheme as above, we
introduce a new variable $t$ and replace any equality of the form $(x=y)$ by $(x=y)\land t$ (resp. $(x=y)\land \neg t$).
\end{proof}

\begin{proposition}\label{prop:arg_rel_NP_complete}
Let $\Gamma$ be a constraint language. If $\Gamma$ is $\Schaefer$ but neither $\positive$ nor $\negative$ then $\ARGREL(\Gamma)$ is $\NP$-complete.
\end{proposition}

\begin{proof}
It remains to show $\NP$-hardness. If $\Gamma$ is $\zerovalid$ and $\onevalid$, we conclude with Lemma~\ref{lem:impl_equality} and Lemma~\ref{lem:hardness_for_eq_argrel}.
If $\Gamma$ is not $\zerovalid$ but $\onevalid$ (resp. $\zerovalid$ but not $\onevalid$), we conclude with Lemma~\ref{lem:impl_equality_with_T} and Lemma~\ref{lem:hardness_for_eq_argrel}.
If $\Gamma$ is neither $\zerovalid$ nor $\onevalid$, $\ARG(\Gamma$) is $\NP$-hard and we conclude with Remark~\ref{rem:useful_remark}.
\end{proof}

To conclude the proof of Theorem \ref{thm:arg_rel}
it remains to deal with the $\SigPtwo$-complete cases.

\begin{proposition}\label{prop:argrel_notschaefer}
Let $\Gamma$ be not $\Schaefer$.
Then $\ARGREL(\Gamma)$ is $\SigPtwo$-complete.
\end{proposition}

\begin{proof}
Membership in $\SigPtwo$ follows as for $\ARG$: given an instance $(\Delta,
\alpha, \psi)$, 
guess a support $\Phi\subseteq\Delta$, verify  that $\psi \in \Phi$, and
subsequently check with an
$\NP$-oracle that   $\Phi$ is consistent, $\Phi$ entails $\alpha$
and that $\Phi$ is minimal w.r.t. the last property. 

We turn to the hardness proof. We will use the problem $\ARGREL(\Gamma', R)$ with
two parameters, in which we differentiate the restrictions put on the
knowledge base and the claim, as an intermediate problem.

If $\Gamma$ is $\complementive$, we can apply the same trick as in Proposition \ref{prop:hard_argcheck_cp}
in order to reduce from the non-$\complementive$ case.
It therefore suffices to show hardness for non-$\complementive$ $\Gamma$.
We perform a case distinction according to whether $\Gamma$ is 0/1-valid or not.
%
%

\medskip

\noindent\emph{$\Gamma$ is neither 1-valid, nor 0-valid.} 
%
For those $\Gamma$ the problem $\ARG(\Gamma$) is $\SigPtwo$-hard. We conclude
with Remark~\ref{rem:useful_remark}.

\medskip

\noindent\emph{$\Gamma$ is both 1-valid and 0-valid.}
%
We give a reduction from the $\SigPtwo$-hard problem $\ABD(\Gamma)$. An instance is
given by $(\varphi, H, q)$, where $\varphi$ is a $\Gamma$-formula, $H \subseteq
\Vars{\varphi}$ and $q$ is a variable. The instance $(\varphi, H, q)$ is a
positive one if and only if there is an $E\subseteq \Lits{H}$ such that
  $\varphi \land E$ is satisfiable and 
 $\varphi \land E \models q$.

We give the following sequence of reductions.
$$\begin{array}{rl}
\ABD(\Gamma) & \leqlogm \ARGREL(\Gamma_1, x \lor \neg y \lor z) \\
									 & \leqlogm \ARGREL(\Gamma_2, (x = y) \land (z = w)) \\
									 & \leqlogm \ARGREL(\Gamma),
\end{array}$$
where
$$\begin{array}{rl}
\Gamma_1\;  & =\; \Gamma \cup \{(x \lor \neg y), (x = y)\},\\
\Gamma_2\;  & =\; \Gamma_1 \cup \{R_\delta\},\\
R_\delta(x_1, \dots, x_7)\;  & =\; \big((x_1 \lor \neg x_2 \lor x_4) \leftrightarrow (x_4 = x_5) \big) \land (x_6 = x_7).\\
\end{array}$$
We will first treat the third and the second reduction which are short and very technical.
Observe that $\Gamma_2$ is both $\zerovalid$ and $\onevalid$. We have therefore by Lemma~\ref{lem:impl_CocloneInclusions}, first item, that $\Gamma_2 \subseteq \clos{\Gamma}$. Since $\Gamma$ is not $\Schaefer$, we have by Lemma~\ref{lem:impl_not_schaefer_equality} that $\Gamma_2 \subseteq \closneq{\Gamma}$. 
Further, we have by Lemma~\ref{lem:impl_equality} that $(x = y) \land (z = w)
\in \closnexneq{\Gamma}$. Therefore, the third reduction follows by
the second item of Lemma~\ref{lem:not_Schaefer_corr}.

For the second reduction let $(\Delta,\alpha,\psi)$ be an instance of
$\ARGREL(\Gamma_1, x \lor \neg y \lor z)$, where $\alpha = (x_\alpha \lor \neg
y_\alpha \lor z_\alpha)$.
%
%
We construct the instance $(\Delta',\alpha',\psi')$ of $\ARGREL(\Gamma_2, (x =
y) \land (z = w))$ as follows. Let $u_1,u_2,v_1,v_2$ be fresh variables. Then
we define:
$$\begin{array}{rl}
\Delta'  =\; & \Delta \cup \{\delta\}\\
\alpha'  =\; & (u_1=u_2) \land (v_1=v_2)\\
\psi'    =\; & \psi \\
\delta   =\; & R_\delta(x_\alpha, y_\alpha, z_\alpha, u_1, u_2, v_1, v_2)\\
\end{array}$$
We observe that $\Delta'$ is a set of $\Gamma_2$-formul\ae, as desired.
By definition of $R_\delta$, the formula $\delta$ is equivalent
to $\big(\alpha \leftrightarrow (u_1=u_2)\big) \land (v_1=v_2)$. This allows us
to observe that
any support for $\alpha'$ will contain the formula $\delta$ which
assures a one-to-one correspondence between the supports of the two instances.

\medskip

\noindent
It remains to give the first reduction which constitutes the main transformation idea between $\ABD$ and $\ARGREL$.
Let $(\varphi, H, q)$ be an instance of
$\ABD(\Gamma)$, where $H = \{h_1, \dots, h_k\}$. We construct the instance
$(\Delta, \alpha, \psi)$ of $\ARGREL(\Gamma_1, x \lor \neg y \lor z)$ as
follows. Let $s,t,f$ be fresh variables. Then
we define:
$$\begin{array}{rl}
\Delta  =\;   & \{(h_i \lor \neg t), (\neg h_i \lor f) \mid 1 \leq i \leq n\} 
			 \cup \{\varphi\}
			 \cup \{\psi\}\\
\alpha  =\;   & (s \lor \neg t \lor f)\\
\psi    =\;   & (s = q) \\
\end{array}$$
We observe that $\Delta$ is a set of $\Gamma_1$-formul\ae, as desired.

We now show that there is an explanation for $(\varphi,H,q)$ if and only if
$\Delta$ contains a minimal support for $\alpha$ containing $\psi$.
For the left-to-right implication let $E\subseteq \Lits{H}$ such that
 $\varphi \land E$ is satisfiable and 
  $\varphi \land E \models q$.

We define
$$\begin{array}{rl}
\Phi\;\;   =\;\;    & \{(h_i \lor \neg t) \mid h_i \in E\} 
             \cup  \{(\neg h_i \lor f) \mid \neg h_i \in E\} 
						 \cup \{\varphi\} 
						 \cup  \{\psi\}.
\end{array}$$
Note that it suffices to show that 

\medskip

\begin{enumerate}\setlength{\itemsep}{\itemseplength}
	\item[a)] $\psi \in \Phi$,
	\item[b)] $\Phi$ is satisfiable,
	\item[c)] $\Phi \models \alpha$, and
	\item[d)] $\Phi \backslash \{\psi\} \not \models \alpha$.
\end{enumerate}

\medskip

\noindent
Such a support is not necessarily minimal, but will contain a minimal support as
desired.
Item a) holds by construction of $\Phi$ and item b) follows from the
assumption that all formul\ae\ are $\zerovalid$ and
$\onevalid$.

We turn to item c). It suffices to show the following four cases.

\medskip

\begin{itemize}\setlength{\itemsep}{\itemseplength}
	\item $\Phi[t/0,f/0] \models \alpha[t/0,f/0]$
	\item $\Phi[t/0,f/1] \models \alpha[t/0,f/1]$
	\item $\Phi[t/1,f/1] \models \alpha[t/1,f/1]$
	\item $\Phi[t/1,f/0] \models \alpha[t/1,f/0]$
\end{itemize}

\medskip

\noindent
The first three cases are obvious, since $$\alpha[t/0,f/0] \equiv
\alpha[t/0,f/1]
\equiv \alpha[t/1,f/1] \equiv 1.$$
In order to show the last one, observe that $\Phi[t/1,f/0] \equiv \varphi \land
E \land (s = q)$ and that $\alpha[t/1,f/0] \equiv s$. This shows that 
$\Phi[t/1,f/0] \models \alpha[t/1,f/0]$ since $\varphi\land E\models q$.

We turn to item d). It suffices to show that $(\Phi\backslash\{\psi\})[t/1,f/0]
\not \models \alpha[t/1,f/0]$. But this is obvious, since
$(\Phi\backslash\{\psi\})[t/1,f/0] \equiv \varphi \land E$ and  $\alpha[t/1,f/0]
\equiv s$ and $s \notin \Vars{\varphi \land E}$.

We now turn to the right-to-left implication. Let $\Phi \subseteq \Delta$ such
that

\medskip

\begin{enumerate}\setlength{\itemsep}{\itemseplength}
	\item[a)] $\psi \in \Phi$,
	\item[b)] $\Phi$ is satisfiable,
	\item[c)] $\Phi \models \alpha$, and
	\item[d)] $\Phi \backslash \{\psi\} \not \models \alpha$.
\end{enumerate}

\medskip

\noindent
We define
$$\begin{array}{rl}
E\;\;   =\;\;     & \{h_i \mid (h_i \lor \neg t) \in \Phi \} 
         \cup \{\neg h_i \mid (\neg h_i \lor f) \in \Phi\}.
\end{array}$$
We first show that $\varphi \land E$ is satisfiable (this implies also that $E$
is satisfiable).
From d) we know that $\Phi\backslash\{\psi\} \land \neg \alpha$ is satisfiable.
That is, $\Phi\backslash\{\psi\} \land \neg s \land t \land \neg f$ is
satisfiable.
We conclude that in particular the formula $(\Phi\backslash\{\psi\})[t/1,f/0]$
is satisfiable, which is equivalent to $\varphi \land E$.

From c) we conclude that in particular $\Phi[t/1,f/0] \models \alpha[t/1,f/0]$.
The formula $\alpha[t/1,f/0]$ is equivalent to $s$. By definition of $E$ and
since $\psi \in \Phi$, the formula $\Phi[t/1,f/0]$ is either equivalent to $E
\land \varphi \land \psi$ or to $E \land \psi$. Since $E \land \psi \equiv E
\land (s = q)$ and $s,q \notin E$ we have that $E \land \psi \not \models s$.
Thus the first case applies.
That is, we have that $\varphi \land E \land (s = q) \models s$.
Since $s \notin \Vars{\varphi \land E}$ we conclude that $\varphi \land E
\models q$.

\medskip

\noindent\emph{$\Gamma$ is 1-valid but not 0-valid (or the converse).}
%
By duality it suffices to treat the case where $\Gamma$ is $\onevalid$ and not
$\zerovalid$.
We give a reduction from the $\SigPtwo$-hard problem $\ABD(\Gamma)$. The
structure of the proof is the same as in the previous case.
We give the following sequence of reductions. 
$$\begin{array}{rl}
\ABD(\Gamma) & \leqlogm \ARGREL(\Gamma_1, x \lor y) \\
									 & \leqlogm \ARGREL(\Gamma_2, x \land y) \\
									 & \leqlogm \ARGREL(\Gamma),
\end{array}$$
where
$$\begin{array}{rl}
\Gamma_1\;  & =\; \Gamma \cup \{\T, (x = y), (x \lor \neg y)\},\\
\Gamma_2\;  & =\; \Gamma_1 \cup \{R_\delta\},\\
R_\delta(x_1, \dots, x_4)\;  & =\; \big((x_1 \lor x_2) \leftrightarrow x_3 \big) \land x_4.\\
\end{array}$$
Since $\Gamma_2$ is $\onevalid$, we have by Lemma~\ref{lem:impl_CocloneInclusions}, second item that $\Gamma_2 \subseteq \clos{\Gamma}$. Since $\Gamma$ is not $\Schaefer$, we have by Lemma~\ref{lem:impl_not_schaefer_equality} that $\Gamma_2 \subseteq \closneq{\Gamma}$. 
Further, we have by Lemma~\ref{lem:implementations}, fourth item that $x \land y
\in \closnexneq{\Gamma}$. Therefore, the third reduction follows by
the second item of Lemma~\ref{lem:not_Schaefer_corr}.

For the second reduction let $(\Delta,\alpha,\psi)$ be an instance of
$\ARGREL(\Gamma_1, x \lor y)$, where $\alpha = (x_\alpha \lor y_\alpha)$.
We construct the instance $(\Delta',\alpha',\psi')$ of $\ARGREL(\Gamma_2, x
\land y)$ as follows. Let $u,v$ be fresh variables. Then we define:
$$\begin{array}{rl}
\Delta'  =\; & \Delta \cup \{\delta\}\\
\alpha'  =\; & u \land v\\
\psi'    =\; & \psi \\
\delta   =\; & R_\delta(x_\alpha, y_\alpha, u, v)\\
\end{array}$$
We observe that $\Delta'$ is a set of $\Gamma_2$-formul\ae, as desired.
By definition of $R_\delta$, the formula $\delta$ is equivalent
to $\big(\alpha \leftrightarrow u\big) \land v$. This allows us to observe that
any support for $\alpha'$ will contain the formula $\delta$ which
assures a one-to-one correspondence between the supports of the two instances.

\medskip

For the first reduction let $(\varphi, H, q)$ be an instance of
$\ABD(\Gamma)$, where $H = \{h_1, \dots, h_k\}$. We construct the instance
$(\Delta, \alpha, \psi)$ of $\ARGREL(\Gamma_1, x \lor y)$ as follows. Let $s,f$
be fresh variables. Then
we define:
$$\begin{array}{rl}
\Delta  =\; & \{h_i, (\neg h_i \lor f) \mid 1 \leq i \leq n\}  \cup
\{\varphi\} \cup \{\psi\}\\
\alpha  =\; & s \lor f\\
\psi    =\; & (s = q) \\
\end{array}$$
Obviously, $\Delta$ is a set of $\Gamma_1$-formul\ae, as desired.
Correctness can be proved as in the previous case, though more easily.
\end{proof}


%
%

\section{Conclusion and future work}\label{sec:conclusion}

In this paper we presented complete complexity classifications for three important
computational tasks in argumentation, namely the existence, the verification and the
relevance problem. The classifications have been obtained in 
 Schaefer's popular framework, i.e., formul\ae\ are in generalized conjunctive
normal form and  restrictions are made on the allowed type of constraints
(generalized clauses).
This approach covers classical classes of CNF-formul\ae. For
instance we obtain that the argument existence problem
is  $\NP$-complete for  
$\horn$-, $\dualhorn$-, $\affine$- and 2CNF-formul\ae, whereas the argument
verification problem is tractable in these cases. Observe that the frontier  between hard and easy problems for
$\ARGCHECK$ is the same as for the implication problem $\IMP$.
It may come as a surprise that there are fragments
(for instance in the case of 0-valid non-Schaefer
relations) for which verifying an argument is potentially harder than deciding the existence of an 
argument ($\ARGCHECK$ is $\DP$-complete, $\ARG$ is ``only'' $\coNP$-complete).
Finally, the relevance problem is the hardest among the studied problems: already the equality relation makes it $\NP$-hard. The only tractable fragment is that of positive/negative formul\ae.


It would be interesting to extend the study to different variants on the
claim as it has been done in \cite{NoZa08,CrScTh11} for abduction.
The complexity of 
the problems studied in this paper
is also a computational core for
evaluating more complex argumentation problems, for instance, the warranted formula 
problem (WFP) on argument trees, which has  been shown to be $\PSPACE$-complete
\cite{GoHi10}.
It might be the case  that fragments 
studied here also lower the complexity of WFP, but we  leave details for future
work.

\subsection*{Acknowledgements}
The authors would like to thank Julian-Steffen M{\"u}ller for the nice proof of Proposition~\ref{lem:hardness_for_eq_argrel}.

\bibliographystyle{alpha}
\bibliography{\jobname}

\end{document}